\documentclass{article}

\usepackage{amsmath,amsthm,amsfonts,braket,graphicx}

\theoremstyle{definition} 
\theoremstyle{definition} 
\newtheorem {theorem} {Theorem}

\newtheorem {lemma} {Lemma}

\title{Mediated Semi-Quantum Key Distribution}
\author{Walter O. Krawec\\\small{Stevens Institute of Technology}\\\small{Hoboken NJ, 07030 USA}\\\small{\texttt{walter.krawec@gmail.com}}}

\begin{document}
\maketitle
\begin{abstract}
In this paper, we design a new quantum key distribution protocol, allowing two limited semi-quantum or ``classical'' users to establish a shared secret key with the help of a fully quantum server.  A semi-quantum user can only prepare and measure qubits in the computational basis and so must rely on this quantum server to produce qubits in alternative bases and also to perform alternative measurements.  However, we assume that the sever is untrusted and we prove the unconditional security of our protocol even in the worst case: when this quantum server is an all-powerful adversary.  We also compute a lower bound of the key rate of our protocol, in the asymptotic scenario, as a function of the observed error rate in the channel allowing us to compute the maximally tolerated error of our protocol.  Our results show that a semi-quantum protocol may hold similar security to a fully quantum one.
\end{abstract}


\section{Introduction}

Quantum Key Distribution (QKD) protocols (see \cite{QKD-survey} for a general survey) allow two parties, customarily referred to as Alice ($A$) and Bob ($B$), to agree on a shared secret key (a string of random classical bits) even in the presence of an all-powerful adversary: Eve ($E$).  Semi-Quantum Key Distribution (SQKD), first introduced in \cite{SQKD-first}, involves the construction of QKD  protocols whereby one party (typically $B$) is only allowed to measure and prepare qubits in the computational $Z$ basis (written here $\ket{0}, \ket{1}$).  Due to this limitation, $B$ is commonly called a ``classical Bob''.  Alice, however, is allowed to prepare and measure in either the $Z$ or $X$ basis (the latter being those states $\ket{+} = \frac{1}{\sqrt{2}}(\ket{0} + \ket{1})$ and $\ket{-} = \frac{1}{\sqrt{2}}(\ket{0} - \ket{1})$).  Such protocols are of great theoretical interest as they help to ascertain exactly how ``quantum'' a protocol need be in order to gain an advantage over its classical counterpart \cite{SQKD-first,SQKD-second}.

In more detail, an SQKD protocol operates over a two-way quantum communication channel which allows a qubit to travel from $A$, to $B$, then back to $A$ (passing through the adversary $E$ twice).  On each iteration of such a protocol, $A$ prepares a qubit in either the $Z$ or $X$ basis and sends it to $B$ who may then perform one of the following actions:

\begin{enumerate}
  \item \textbf{Measure and Resend}: $B$ may measure the incoming qubit in the $Z$ basis recording the result.  If his measurement result was $\ket{r}$ (for $r \in \{0,1\}$), he then sends a new qubit: $\ket{r}$ back to $A$.
  \item \textbf{Reflect}: Alternatively, $B$ may simply ``reflect'' the qubit back to $A$ without otherwise disturbing it (or learning anything about its state).
\end{enumerate}
A party limited to the above operations is typically referred to as ``classical'' or ``semi-quantum''.

After communicating with each other using the quantum channel, $A$ and $B$ each hold a \emph{raw key} which we denote \texttt{info}$_A$ and \texttt{info}$_B$ respectively.  It is hoped that these two strings are highly correlated and that an adversary has limited information on either.  Following the quantum communication stage, error correction (EC) and privacy amplification (PA) are performed to distill a (smaller) secure secret key.  The reader is referred to \cite{QKD-survey} for more information on these, now standard, processes.

In this paper we design a new SQKD protocol which allows two semi-quantum/``classical'' users $A$ and $B$ to distill a secure key with the help of an untrusted, fully quantum center/server $C$ (see Figure \ref{fig:schematic} for a diagram of this scenario).  Such a protocol we call a \emph{mediated semi-quantum key distribution protocol} (it is also a type of \emph{multi-user quantum key distribution protocol} - we will use the two terms interchangeably).  In this scenario, $C$ will prepare quantum states, and forward them to $A$ and $B$.  Since $A$ and $B$ may only reflect or measure in the $Z$ basis, they must rely on $C$ to perform measurements in alternative bases to verify the security of the quantum channel.  We assume, however, that, not only may there be an independent, all-powerful adversary $E$ listening on the quantum channel, but also that $C$ may be an adversary (see Figure \ref{fig:schematic}).  Despite this, we will show our protocol is unconditionally secure by computing the key-rate of our protocol in the asymptotic scenario \cite{QKD-survey}, against an adversarial center and discuss its maximally tolerated error rate.  While other semi-quantum protocols have been shown robust \cite{SQKD-first}, or secure against individual attacks only \cite{SQKD-information}, we are not aware of any similar computation yet done for other semi-quantum protocols (mediated or otherwise); thus our work demonstrates for the first time, that a semi-quantum key distribution protocol may hold comparable security to that of a fully quantum one.

\begin{figure}
\includegraphics[width=\linewidth]{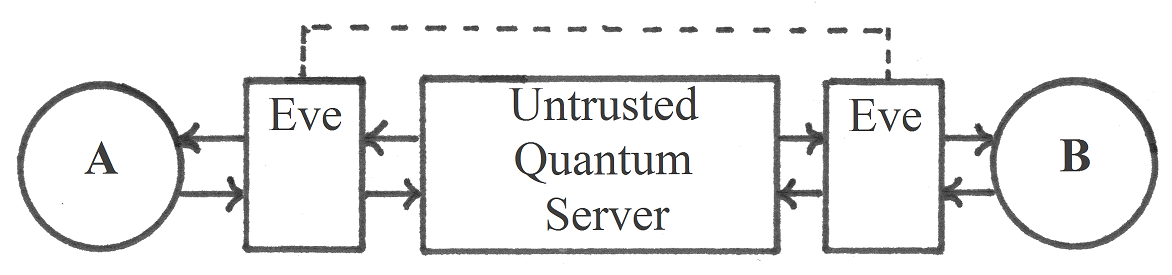}
\caption{A mediated semi-quantum key distribution protocol.  Here, $A$ and $B$ are the two semi-quantum users who wish to distill a secure secret key known only to themselves.  These two users are connected to an untrusted, fully quantum server, by a two-way quantum channel.  The sever is able to send qubits to $A$ and $B$, who may then both return a qubit to the sever.  Not only is the server considered to be adversarial in our work, but there may also exist an eavesdropper sitting between one or both of the channels connecting the sever to $A$ or $B$.  The dashed line connecting the two eavesdroppers in this figure represents a possible connection, either quantum, classical, or physical (in case they are located in the same physical location - perhaps at a branch point in the quantum channel) between the eavesdroppers.}\label{fig:schematic}
\end{figure}

\section{Related Work}

The first multi-user QKD protocol that we are aware of was developed by Phoenix et al. \cite{QKD-Multi-first} and it involved the center $C$ sending Bell states, with one particle sent to each party.  Encoding was done by $A$ and $B$ performing unitary operations on their respective qubits and returning the result to the center, who performed a measurement, informing $A$ and $B$ of the result.

In \cite{QKD-Multi-with-memory}, a system was described whereby $A$ and $B$ send particles, prepared in one of the states $\ket{0}, \ket{1}, \ket{+}$, or $\ket{-}$, to a quantum center who must then store them in a quantum memory until such future time when a secret key is desired.  The protocols of \cite{QKD-Multi-second,QKD-Multi-third,QKD-Multi-fifth-comment,QKD-Multi-authenticate} required the users to measure in both $Z$ and $X$ bases, while the protocols of \cite{QKD-Multi-third,QKD-Multi-fifth} require $A$ and $B$ to apply various unitary operators to the qubits arriving to them from the center.

Preparation of Bell states was required of both users in the protocol of \cite{QKD-Multi-fourth}; here the center was used to perform Bell measurements (this protocol was later improved in \cite{QKD-Multi-fourth-comment,QKD-Multi-fourth-comment2} forcing the two users to also measure in both bases thus avoiding a security flaw whereby the center was able to learn the key).  In \cite{QKD-Multi-fifth} (later improved in \cite{QKD-Multi-fifth-comment}), the center was required to prepare states involving the entanglement of four qubits.  Furthermore, one of the users was required to perform a Bell measurement.

On the semi-quantum side, the work most similar to ours are the protocols of \cite{SQKD-multi1,SQKD-cl-A}.  The protocol described in \cite{SQKD-multi1} permits a fully quantum $A$ to share a key with a series of semi-quantum users $B_1, B_2, \cdots, B_n$.  However, the fully-quantum center $A$ is also fully-trusted (since $A$ also desires to share the key with each $B_i$).  The protocol described in \cite{SQKD-cl-A}, like our protocol, permits two classical users to distill a secure secret key with the help of an untrusted quantum server; however it also required the existence of a private quantum channel, which the quantum server could not access (though, indeed, $E$ is allowed access), connecting the two users.

\section{The Protocol}

Denote by $A$ and $B$ the two semi-quantum users who wish to distill a shared, secure secret key, and by $C$ the untrusted quantum center/server.  We assume that there is a quantum communication channel connecting $A$ to $C$ and $B$ to $C$.  We also assume there is an authenticated classical channel connecting $A$ and $B$.  Anyone may listen on this classical channel, however only $A$ and $B$ may write to it.  There is a classical channel connecting $C$ to either $A$ or $B$.  Any message that $C$ sends to $A$ (or $B$) is ``relayed'' by that recipient, over the authenticated channel, to $B$ (or $A$).  Thus, $C$ cannot send different messages to each of the two users without it being known.

Whether the channel connecting $C$ to the semi-quantum users needs to be authenticated depends on the application.  In the next section, where we prove the security of the protocol, we assume that $C$ is adversarial; however we assume, at first, that there are no third party eavesdroppers.  In this scenario, there is clearly no need for $C$'s messages to be authenticated.  Later, we consider the worst case: $C$ is adversarial, and there are other, third party eavesdroppers.  However, in this case, the third party eavesdroppers' attacks (including any attack which attempts to impersonate $C$ on the classical channel and alter his messages) may be ``absorbed'' into an attack considered in the first case; thus, our analysis is exactly the same, and there is no need, as far as the ``worst-case'' security bounds we prove are concerned, for the channel from $C$ to be authenticated (the channel between $A$ and $B$, however, does need to be authenticated).  Of course, if the quantum server is, for example, a payed service, then an authenticated channel might be necessary to prevent these third-party attackers from altering $C$'s messages (and indeed a better security bound may be possible in this case as we mention later).  Note that these issues are not unique to our protocol.

Our protocol utilizes the Bell basis whose states we denote:
\begin{align*}
\ket{\Phi^\pm} &= \frac{1}{\sqrt{2}}(\ket{00} \pm \ket{11})\\
\ket{\Psi^\pm} &= \frac{1}{\sqrt{2}}(\ket{01} \pm \ket{10})\\
\end{align*}

We will first describe our protocol assuming a perfectly honest and trustworthy center.  A single iteration $i$ of our protocol consists of the following process:

\begin{enumerate}
  \item $C$ prepares the Bell state $\ket{\Phi^+}$.  The center will then send one particle to $A$ and the other to $B$.
  \item $A$ ($B$) will choose randomly to reflect the received qubit, with probability $p_R^A$ ($p_R^B$), or measure and resend (with probability $p_M^{A/B} = 1-p_R^{A/B}$).  If she (he) measures and resends, she (he) will save the result as $k^A_i$ ($k^B_i$).  Note that $A$ and $B$ choose independently their action and do not yet disclose their choice.
  \item $C$ will receive both particles from $A$ and $B$ and will perform a Bell measurement.  Assuming no other noise, this measurement should produce either $\ket{\Phi^+}$ or $\ket{\Phi^-}$ (if there is noise, the measurement might alternatively produce one of $\ket{\Psi^\pm}$).  She will inform $A$ and $B$ of the measurement result.  For simplicity, $C$ will send to $A$ and $B$ the message ``$-1$'' if the measurement result was $\ket{\Phi^-}$ and ``$+1$'' for any other result.
  \item $A$ and $B$ disclose whether they reflected or measured.  If both $A$ and $B$ reflected, they should expect $C$'s answer to be ``$+1$'' - if they receive ``$-1$'', this round is counted as an error; if the number of errors exceeds some threshold $\tau$ (a security parameter set by $A$ and $B$), both $A$ and $B$ abort.
  \item If both $A$ and $B$ measured and if $C$ reported ``$-1$'', they will keep their raw key bits $k_i^A$ and $k_i^B$, otherwise, if $C$ sent ``$+1$'' they will disregard this iteration.
\end{enumerate}

Note that an iteration is used for the key bit only if $C$ sends ``$-1$''; while decreasing the number of ``useful'' iterations, this turns out to be a vital addition to our protocol.  This process repeats $N$ times for $N$ sufficiently large.  Assuming no noise, it is expected that $Np_M^Ap_M^B/2$ of these iterations contribute to the final raw key $\texttt{info}_A$ and $\texttt{info}_B$ (following \cite{QKD-BB84-Modification}, $A$ and $B$ may set $p_M$ arbitrarily close to $1$).  Next $A$ and $B$ will divulge a randomly chosen subset of their measurement results to estimate certain channel statistics; for example the $Z$ basis error rate (next section we discuss the required parameters that must be estimated).  Following this, $A$ and $B$ will perform error correction and privacy amplification.  This is done independent of $C$, however we assume that the center can listen in on and utilize this information.

\section{Security Against Collective Attacks}

Of course the point of our protocol is to operate assuming an untrusted, adversarial center.  Such a center $C$ we assume will prepare and send any state of his choice in step (1) (however, he will send only a single qubit to each party).  This includes the possibility of preparing a state entangled with his own private ancilla.  Also on step (3), we allow $C$ to perform any operation of his choice (within the laws of physics of course).  This includes the possibility of him keeping, for now, information unmeasured in a private quantum ancilla.  However he must send a message either ``$+1$'' or ``$-1$'' to both $A$ and $B$ (since anything written on the classical channel is known to both $A$ and $B$, we might as well assume $C$ is restricted to sending the same message to both - that is, he cannot send ``$+1$'' to $A$ and ``$-1$'' to $B$ without being caught).  It is $C$'s goal to learn the secret key which $A$ and $B$ eventually distill.

We will consider \emph{collective attacks} in this section.  These are attacks whereby $C$ performs the same attack each iteration; however $C$ is free to keep quantum information in his private ancilla for later measurement.  For example, it might be advantageous for $C$ to choose an optimal POVM after seeing a cipher text, encrypted using $A$ and $B$'s distilled key after EC and PA.  This is stronger than an \emph{individual attack} where $C$ is forced to perform his measurement before the key is used for anything.  Later, we will show security against \emph{general attacks} where $C$ is allowed to do anything he likes each iteration (perhaps altering his attack based on the result of attacking the previous iterations).  For more information on these attack models, the reader is referred to \cite{QKD-survey}.  For this section, we will also assume there is no third party eavesdropper $E$ - this scenario will be dealt with in the next section.

We model the center's strategy as follows: first $C$ will send not $\ket{\Phi^+}$ but instead any arbitrary state $\ket{\psi_0} \in \mathcal{H} := \mathcal{H}_{T_A} \otimes \mathcal{H}_{T_B} \otimes \mathcal{H}_C$.  Here $\mathcal{H}_{T_A}$ and $\mathcal{H}_{T_B}$ are two-dimensional Hilbert spaces used to model the particle sent/transmitted to $A$, respectively $B$; $\mathcal{H}_C$ is $C$'s private ancilla which we assume is finite dimensional.  The sent state, which we may assume, without loss of generality, to be pure, may be written as:

\begin{equation}\label{eq:start-entangled}
\ket{\psi_0} = \sum_{i,j\in\{0,1\}} \alpha_{i,j}\ket{i,j}_{T_A,T_B} \otimes \ket{c_{i,j}}_C.
\end{equation}
Here, each $\alpha_{i,j} \in \mathbb{C}$ and $\sum_{i,j}|\alpha_{i,j}|^2 = 1$.  We may assume that each $\ket{c_{i,j}}$ is normalized, though they are not necessarily orthogonal.

Next, after $A$ and $B$ perform their operations and return the qubits to $C$, we permit $C$ to perform any quantum operation on the returned state.  Since he must also report a single bit of classical information, we model this as a \emph{quantum instrument} \cite{QC-Instrument} $\mathcal{I}$ such that for any density matrix $\rho$ acting on $\mathcal{H}$:

\[
\mathcal{I}(\rho) = \sum_{k=1}^{N_0}E_{k,0}\rho E_{k,0}^* \otimes \ket{+1}\bra{+1}_{cl} + \sum_{k=1}^{N_1}E_{k,1}\rho E_{k,1}^* \otimes \ket{-1}\bra{-1}_{cl},
\]
under the condition that:
\[
\sum_{k=1}^{N_0}E_{k,0}^*E_{k,0} + \sum_{k=1}^{N_1}E_{k,1}^*E_{k,1} = I.
\]
In our case, we may assume $N_0$ and $N_1$ are both finite.

Note we have expanded the underlying Hilbert space, adding the subspace $\mathcal{H}_{cl}$: a two dimensional subspace spanned by orthonormal basis $\{\ket{+1}_{cl}, \ket{-1}_{cl}\}$.  This is used to model the classical message $C$ sends to $A$ and $B$ (either ``$+1$'' or ``$-1$'').  Since we assume $C$ is forced to send the same message to $A$ and $B$ (e.g., he cannot send ``$+1$'' to $A$ and ``$-1$'' to $B$), this is sufficient.

\subsection{Unentangled Initial State}
Our first claim is that $C$ may, without any loss of power, send not $\ket{\psi_0}$ as described above (Equation \ref{eq:start-entangled}), but instead the far simpler state: $\ket{\psi_0'} = \sum_{i,j}\alpha_{i,j}\ket{i,j}$.  That is, $C$ need not entangle the start state with his ancilla at first.  This is similar to a result, concerning the security of single state semi-quantum protocols (not mediated ones), first presented in \cite{SQKD-Single-Security} (Lemma 2.1); and, indeed, we use a similar technique for the proof in this subsection.

Assume $C$ sends the entangled state $\ket{\psi_0}$.  Then, when he receives a reply from both $A$ and $B$, the system is in the mixed state:

\[
\rho = p_R^A\cdot p_R^B\ket{\psi_0}\bra{\psi_0} + p_M\left(\sum_{i,j\in\{0,1\}} |\alpha_{i,j}|^2 \ket{i,j,c_{i,j}}\bra{i,j,c_{i,j}}\right),
\]
where $p_M$ is the probability that one or both of $A$ and $B$ measured and resent.  On the other hand, if $C$ simply sends $\ket{\psi_0'}$, the resulting state is:

\[
\rho' = p_R^Ap_R^B\ket{\psi_0'}\bra{\psi_0'}\otimes\ket{0}\bra{0}_C + p_M\left(\sum_{i,j}|\alpha_{i,j}|^2\ket{i,j}\bra{i,j}\otimes\ket{0}\bra{0}_C\right).
\]

Above, we have assumed, without loss of generality, that $C$'s ancilla is cleared to some zero state $\ket{0}_C$.  Define $V$ to be a unitary operator such that $V\ket{i,j,0} = \ket{i,j,c_{i,j}}$.  This may be done unitarily and it is clear that $V\rho'V^* = \rho$.  Thus any arbitrary, entangled state $\ket{\psi_0}$ which $C$ may desire to send at the start may be ``recreated'' later without any loss to $C$.  Therefore, when analyzing the security of this protocol, we need only consider $C$ sending $\ket{\psi_0'}$ - $C$ may then absorb $V$ into his attack $\mathcal{I}$.  As such, we will forgo writing the ``prime''.

\subsection{Unitary Attack}
Next, using standard techniques (see, for instance, \cite{QC-intro,QC-measurement-book}), we may represent $C$'s quantum instrument as a unitary operator acting on a larger Hilbert space.  That is, instead of applying $\mathcal{I}$ on the returned state, $C$ will apply a unitary operator $U$ followed by a projective measurement on the $\mathcal{H}_{cl}$ subspace.

Define $U_{\mathcal{I}}$ as follows:

\[
U_{\mathcal{I}} := \sum_{k=1}^{N_0} E_{k,0} \otimes \ket{k}_{E_1}\ket{0}_{E_2}\ket{+1}_{cl} + \sum_{k=1}^{N_1} E_{k,1} \otimes \ket{N_0+k}_{E_1}\ket{1}_{E_2}\ket{-1}_{cl}.
\]

It is clear that this is an isometry mapping $\mathcal{H} \rightarrow \mathcal{H} \otimes \mathcal{H}_{E_1} \otimes \mathcal{H}_{E_2} \otimes \mathcal{H}_{cl}$ where $\mathcal{H}_{E_1}$ is spanned by orthonormal basis $\{\ket{i}_{E_1} \text{ } | \text{ } i = 1, 2, \cdots, N_0+N_1\}$ and $\mathcal{H}_{E_2}$ is spanned by orthonormal basis $\{\ket{0}_{E_2}, \ket{1}_{E_2}\}$.  That this is an isometry is clear:

\[
\begin{array}{rccl}
U_{\mathcal{I}}^*U_{\mathcal{I}} &=&& \left(\sum_{k=1}^{N_0} E_{k,0}^* \otimes \bra{k,0,+1} + \sum_{k=1}^{N_1} E_{k,1}^* \otimes \bra{N_0+k,1,-1}\right)\\
&&\cdot&\left(\sum_{k=1}^{N_0} E_{k,0} \otimes \ket{k,0,+1} + \sum_{k=1}^{N_1} E_{k,1} \otimes \ket{N_0+k,1,-1}\right)\\\\
&=&& \sum_{k=1}^{N_0}E_{k,0}^*E_{k,0} \otimes \braket{k,0,+1|k,0,+1}\\
&&+&\sum_{k=1}^{N_1}E_{k,1}^*E_{k,1}\otimes\braket{N_0+k,1,-1|N_0+k,1,-1}\\\\
&=&&\sum_{k=1}^{N_0}E_{k,0}^*E_{k,0} + \sum_{k=1}^{N_1}E_{k,1}^*E_{k,1} = I.
\end{array}
\]

It is also clear that, given $\ket{\psi} \in \mathcal{H}$, then $tr_{E_1,E_2}(U_{\mathcal{I}}\ket{\psi}\bra{\psi}U_{\mathcal{I}}^*) \equiv \mathcal{I}(\ket{\psi}\bra{\psi})$ (this is also the case for any mixed state due to linearity).

We may assume that $C$ has access to $\mathcal{H}_{E_1} \otimes \mathcal{H}_{E_2}$ (this can only add to $C$'s power).  Furthermore, we may assume without loss of generality that when $C$ receives a state $\ket{\psi} \in \mathcal{H}$ from $A$ and $B$ (after they performed their operations), that the complete quantum state which $C$ then holds is $\ket{\psi,0,0,0} \in \mathcal{H}\otimes\mathcal{H}_{E_1}\otimes\mathcal{H}_{E_2}\otimes\mathcal{H}_{cl}$ (where $\ket{0}_{cl}$ is some arbitrary, normalized, zero state in $\mathcal{H}_{cl}$, the choice of which is irrelevant).

Finally, $U_{\mathcal{I}}$ may be extended to a unitary operator $\mathcal{U}_{\mathcal{I}}$, so that $\mathcal{U}_{\mathcal{I}}\ket{\psi}\otimes\ket{0,0,0}_{E_1,E_2,cl} \equiv U_{\mathcal{I}}\ket{\psi}$ in a straight-forward manner (its action on states $\ket{\psi}\otimes\ket{i,j,l}$ for $i,j,l\ne0$ is irrelevant).

Thus:

\[
tr_{E_1,E_2}\left(\mathcal{U}_{\mathcal{I}}\ket{\psi,0,0,0}\bra{\psi,0,0,0}\mathcal{U}_{\mathcal{I}}^*\right) = tr_{E_1,E_2}\left(U_{\mathcal{I}}\ket{\psi}\bra{\psi}U_{\mathcal{I}}^*\right) = \mathcal{I}(\ket{\psi}\bra{\psi}),
\]
and so we may therefore safely consider only unitary operators when discussing $C$'s attack.  The result of $C$ sending the classical message ``$+1$'' is equivalent to projecting the $\mathcal{H}_{cl}$ portion of the resulting state (after applying $\mathcal{U}_\mathcal{I}$) to $\ket{+1}_{cl}$; similarly for sending ``$-1$''.  Thus $C$'s choice/decision as to what message to send, followed by the resulting state of his quantum system, is simply modeled by a projective measurement of the $\mathcal{H}_{cl}$ subspace in the $\{\ket{+1}_{cl}, \ket{-1}_{cl}\}$ basis.  Furthermore, we will simply ``absorb'' the auxiliary spaces $\mathcal{H}_{E_i}$ into $\mathcal{H}_C$ from this point forward.

\subsection{Key Rate}
We may now consider the key rate of our protocol in the asymptotic scenario \cite{QKD-survey}.  That is, after completing $n$ successful iterations of our protocol (those iterations where $A$ and $B$ both measure and $C$ sends ``$-1$'' - thus they each have an \texttt{info} string/raw key of $n$ bits), $A$ and $B$ will, as is usual with QKD protocols, perform an error correcting (EC) routine followed by privacy amplification (PA) to distill a secure secret key of size $l(n) \le n$.  Here $l(n)$ depends on the amount of information $C$ may hold on the key.  The key rate, then, is defined as:

\begin{equation}\label{eq:keyrate}
r := \lim_{n \rightarrow \infty} \frac{l(n)}{n}.
\end{equation}

Under collective attacks, one may model the conclusion of $n$ successful iterations of our protocol, before EC and PA, with the following density matrix:

\begin{equation}\label{eq:cqq-state}
\rho_{ABC} = \sum_{x,y\in\{0,1\}} p(x,y) \ket{x,y}\bra{x,y} \otimes \rho_C^{(x,y)},
\end{equation}
where $p(x,y)$ is the probability of $A$ and $B$'s raw key being $x$ and $y$ respectively, while $\rho_C^{(x,y)}$ is the state of $C$'s ancilla in this event.

Using this terminology, we may bound the key-rate $r$ using Devetak-Winter's equation \cite{QKD-Winter-Keyrate}:
\begin{equation}
r \ge I(A:B) - I(A:C)
\end{equation}
where $I(A:C) = S(\rho_A) + S(\rho_C) - S(\rho_{AC})$ is the quantum mutual information held between $A$ and $C$ (here we use $S(\cdot)$ to denote the von Neumann entropy).  Similarly, $I(A:B)$ is the mutual information between $A$ and $B$, though this computation uses the classical Shannon entropy function; that is $I(A:B) = H(A) + H(B) - H(A,B)$.  For the states in question, namely those of the form described by Equation \ref{eq:cqq-state}, it is not difficult to see that $I(A:C) = S(\rho_C) - \sum_x p(x)S(\rho_C^{(x)})$; note that here we have traced out $B$ from the original state.  See Ref \cite{QKD-Winter-Keyrate} for more details.  Observe, however, that if $r > 0$, $A$ and $B$ may, after EC and PA, distill a secure secret key.

After a particular run of a QKD protocol, the quantity $I(A:B)$ is easily computed; the difficulty lies in the computation of $I(A:C)$.  Of course, if $A$ and $B$ have access only to certain statistics and not the entire attack description (that is, they observe certain probabilities of various events, but they do not know the actual attack operator $U$ that was employed against them and thus cannot explicitly construct Equation \ref{eq:cqq-state}), the problem is exacerbated.  However, they may estimate the key rate:
\begin{equation}\label{eq:keyrate2}
r \ge I(A:B) - \sup\{I(A:C)\},
\end{equation}
where the supremum is taken over all possible attack operators $C$ may employ, which conform the observed statistics \cite{QKD-renner-keyrate,QKD-renner-keyrate2}.  Computing an upper-bound for $I(A:C)$ will be our goal in the next subsection.

\subsection{An Upper-bound on $I(A:C)$}

Given a matrix $A$, we denote its trace norm by $||A||$.  If $A$ is an $n\times n$ Hermitian matrix and $\{\lambda_i\}_{i=1}^n$ are the eigenvalues of $A$, then $||A|| = \sum_{i=1}^n|\lambda_i|$.  An important result, which we will utilize later, was shown in \cite{QC-info-trace-bound} which upper-bounds the quantity:

\begin{equation}
S(\rho_C) - \frac{1}{2}S\left(\rho_C^{(0)}\right) - \frac{1}{2}S\left(\rho_C^{(1)}\right) \le \frac{1}{2}\left|\left|\rho_C^{(0)} - \rho_C^{(1)}\right|\right|.
\end{equation}

In our case, this implies:

\begin{equation}\label{eq:bound-info}
I(A:C) = S(\rho_C) - \frac{1}{2}S\left(\rho_C^{(0)}\right) - \frac{1}{2}S\left(\rho_C^{(1)}\right) \le \frac{1}{2}\left|\left|\rho_C^{(0)} - \rho_C^{(1)}\right|\right|.
\end{equation}

Before continuing, we will require a small lemma which, though fairly trivial, we include for completeness.

\begin{lemma}\label{lemma:eigenvalues}
Let $\rho = \ket{a}\bra{b} + \ket{b}\bra{a}$ where $\ket{a}$ and $\ket{b}$ are arbitrary elements of a $D$-dimensional Hilbert space $\mathcal{H}$ ($D < \infty$) such that $tr\rho = 0$.  Then:
\[
||\rho|| \le 2\sqrt{\braket{a|a}\braket{b|b}}.
\]
\end{lemma}
\begin{proof}
Note that $\rho$ is Hermitian and so its trace norm is simply the sum of the absolute values of its eigenvalues.  We may write $\ket{a} = \alpha\ket{\tilde{a}}$ and $\ket{b} = x\ket{\tilde{a}} + y\ket{\zeta}$ where $\alpha, x, y \in \mathbb{C}$, $\braket{\tilde{a}|\tilde{a}} = \braket{\zeta|\zeta} = 1$ and $\braket{\tilde{a}|\zeta} = 0$.  Furthermore, it holds that $|\alpha|^2 = \braket{a|a}$, $|x|^2 + |y|^2 = \braket{b|b}$ and $\braket{a|b} = \alpha^*x$ (here, if $z \in \mathbb{C}$ we use $z^*$ to denote the conjugate of $z$).  Note that $tr\rho = 0 \Rightarrow \braket{a|b} + \braket{b|a} = \alpha^*x + x^*\alpha = 0$.

Recall that the trace norm is invariant under unitary changes of basis (i.e. $||U\rho U^*|| = ||\rho||$).  Thus, choosing a suitable basis, the first two entries of which are $\{\ket{\tilde{a}}, \ket{\zeta}\}$, we may write:
\[
\rho \equiv \left(\begin{array}{cc}
\sigma&0\\
0&0
\end{array}\right),
\]
where $\sigma$ is a $2\times 2$ matrix:
\[
\sigma = \left(\begin{array}{cc}
0&\alpha y^*\\
\alpha^* y & 0
\end{array}\right).
\]

We see now there are at most two non-zero eigenvalues: they are $\lambda_\pm = \pm\sqrt{|\alpha|^2|y|^2} = \pm\sqrt{\braket{a|a}(\braket{b|b}-|x|^2)}$.  Thus:

\begin{equation}
||\rho|| = |\lambda_-| + |\lambda_+| = 2\lambda_+ = 2\sqrt{\braket{a|a}(\braket{b|b} - |x|^2)} \le 2\sqrt{\braket{a|a}\braket{b|b}}.
\end{equation}
The inequality follows from the fact that both $\braket{a|a}$ and $|x|^2$ are non-negative.
\end{proof}

We now return to our protocol.  Let $p_{i,j}$ be the probability that $A$ and $B$ measure $\ket{i,j}$.  This is a statistic which may be estimated and, as in \cite{QKD-keyrate-general} (where the security of the B92 \cite{QKD-B92} protocol was being proved), we may assume that $A$ and $B$ abort if $p_{00} \ne p_{11}$ or $p_{01} \ne p_{10}$.  That is, we force $C$ to use a ``symmetric'' attack with $p_{00} = p_{11} = (1-Q)/2$ and $p_{01} = p_{10} = Q/2$; thus $Q$ represents the probability that $A$ and $B$'s measurement results differ.

Let $U = \mathcal{U}_{\mathcal{I}}$ denote $C$'s attack operator.  From our earlier discussion, we may assume the state $C$ sends is not entangled with his private ancilla and so we need only consider $U$'s action on $\mathcal{H}_{T_A}\otimes\mathcal{H}_{T_B}$:

\begin{align}
U\ket{\Phi^+} &= \ket{e_0}\ket{+1} + \ket{f_0}\ket{-1}\label{eq:C-attack}\\
U\ket{\Phi^-} &= \ket{e_1}\ket{+1} + \ket{f_1}\ket{-1}\notag\\
U\ket{\Psi^+} &= \ket{e_2}\ket{+1} + \ket{f_2}\ket{-1}\notag\\
U\ket{\Psi^-} &= \ket{e_3}\ket{+1} + \ket{f_3}\ket{-1}\notag,
\end{align}
where the $\ket{e_i}$ and $\ket{f_i}$ are states, not necessarily normalized or orthogonal, in $\mathcal{H}_{T_A}\otimes\mathcal{H}_{T_B}\otimes\mathcal{H}_C$.  Above, we abused notation slightly: since we assume $C$'s system is cleared to some known ``zero'' state, we did not write it on the left-hand side of the above equation.  That is, $U$ is actually acting on $\ket{\Phi^+}\otimes\ket{0}_{C,cl}$ and so on.

By linearity, it follows that:

\begin{align*}
U\ket{00} &= \frac{1}{\sqrt{2}}(\ket{e_0} + \ket{e_1})\ket{+1} + \frac{1}{\sqrt{2}}(\ket{f_0} + \ket{f_1})\ket{-1}\\
U\ket{11} &= \frac{1}{\sqrt{2}}(\ket{e_0} - \ket{e_1})\ket{+1} + \frac{1}{\sqrt{2}}(\ket{f_0} - \ket{f_1})\ket{-1}\\
U\ket{01} &= \frac{1}{\sqrt{2}}(\ket{e_2} + \ket{e_3})\ket{+1} + \frac{1}{\sqrt{2}}(\ket{f_2} + \ket{f_3})\ket{-1}\\
U\ket{10} &= \frac{1}{\sqrt{2}}(\ket{e_2} - \ket{e_3})\ket{+1} + \frac{1}{\sqrt{2}}(\ket{f_2} - \ket{f_3})\ket{-1}.
\end{align*}

We make one additional assumption concerning the ``symmetry'' of $C$'s attack.  Besides forcing $p_{0,0} = p_{1,1}$ and $p_{1,0} = p_{0,1}$, we will have $A$ and $B$ enforce the condition that the probability of $C$ sending ``$-1$'' in the event both $A$ and $B$ measure $0$ is equal to the probability of $C$ sending ``$-1$'' in case $A$ and $B$ both measure $1$ (similarly for the case when $A$ and $B$'s measurements are different).  This statistic can easily be estimated by $A$ and $B$ and, in the asymptotic scenario, our condition is easily enforced.  Note that this forces $Re\braket{f_0|f_1} = Re\braket{f_2|f_3} = 0$.  Later, we will show an example which justifies these assumptions.

Finally, denote by $p_{a}$, the probability that $C$ sends ``$-1$'' on any particular iteration given that both $A$ and $B$ measured.  Given our symmetric attack assumptions above, this value is clearly:
\begin{equation}\label{eq:pra}
p_{a} = \frac{1}{2}(1-Q)(\braket{f_0|f_0} + \braket{f_1|f_1}) + \frac{1}{2}Q(\braket{f_2|f_2} + \braket{f_3|f_3}).
\end{equation}

We are now ready to prove an upper bound for $I(A:C)$.

\begin{theorem}\label{thm:main}
Given the above notation and assumptions, then:
\begin{align}
I(A:C) &\le\frac{(1-Q)}{p_a}\sqrt{\braket{f_0|f_0}\braket{f_1|f_1}} + \frac{Q}{p_a}\sqrt{\braket{f_2|f_2}\braket{f_3|f_3}}\label{eq:first-bound}\\
&\le\frac{(1-Q)}{p_a}\sqrt{\braket{f_0|f_0}} + \frac{Q}{p_a}.\label{eq:second-bound}
\end{align}
\end{theorem}
\begin{proof}
Let $\rho_{ABC}$ denote the state of the quantum system assuming that $A$ and $B$ both measured and resent (otherwise the iteration is discarded and $C$ learns no information), after $C$ applied $U$, but before he performs any measurement on $\mathcal{H}_{cl}$.  Clearly we have:

\[
\rho_{ABC} = \left(\frac{1-Q}{2}\right)(U_{00} + U_{11}) + \left(\frac{Q}{2}\right)(U_{01} + U_{10}),
\]
where $U_{ij} = \ket{i,j}\bra{i,j}_{{A}{B}} \otimes U\ket{i,j}\bra{i,j}U^*$ (here $\ket{i,j}\bra{i,j}_{{A}{B}}$ are $A$ and $B$'s private registers storing their measurement results for this iteration).  Assume that $C$ measures and sends $\ket{-1}_{cl}$ (otherwise the iteration is discarded and there is nothing to learn).  Projecting $\rho_{ABC}$ to $\ket{-1}\bra{-1}_{cl}$ yields:

\begin{align*}
\rho'_{ABC} &= \frac{1}{p}\left(\left(\frac{1-Q}{2}\right)\ket{0,0}\bra{0,0}_{AB}\otimes P(\ket{f_0} + \ket{f_1})\right)\\
&+\frac{1}{p}\left(\left(\frac{1-Q}{2}\right)\ket{1,1}\bra{1,1}_{AB}\otimes P(\ket{f_0} - \ket{f_1})\right)\\
&+\frac{1}{p}\left(\frac{Q}{2}\ket{0,1}\bra{0,1}_{AB}\otimes P(\ket{f_2} + \ket{f_3})\right)\\
&+\frac{1}{p}\left(\frac{Q}{2}\ket{1,0}\bra{1,0}_{AB}\otimes P(\ket{f_2} - \ket{f_3})\right),
\end{align*}
where $P(z) := zz^*$ for any vector $z$ and $p = 2p_a$ (see Equation \ref{eq:pra}).

Tracing out $B$ yields:
\begin{align*}
\rho'_{AC} &= \frac{1}{2}\ket{0}\bra{0}_A \otimes \overbrace{\left(\frac{1-Q}{p}P(\ket{f_0} + \ket{f_1}) + \frac{Q}{p}P(\ket{f_2} + \ket{f_3})\right)}^{\rho_C^{(0)}}\\
&+\frac{1}{2}\ket{1}\bra{1}_A \otimes \underbrace{\left(\frac{1-Q}{p}P(\ket{f_0} - \ket{f_1}) + \frac{Q}{p}P(\ket{f_2} - \ket{f_3})\right)}_{\rho_C^{(1)}}.
\end{align*}

By our assumptions, $Re\braket{f_0|f_1} = Re\braket{f_2|f_3} = 0$ which implies $tr\rho_C^{(0)} = tr\rho_C^{(1)} = 1$.

Let $\widetilde{\rho} = \rho_C^{(0)} - \rho_C^{(1)}$.  Clearly we have:

\begin{equation}
\widetilde{\rho} = \frac{2}{p}((1-Q)\ket{f_0}\bra{f_1} + (1-Q)\ket{f_1}\bra{f_0} + Q\ket{f_2}\bra{f_3} + Q\ket{f_3}\bra{f_2}).
\end{equation}

Let $T := \frac{1}{2}||\widetilde{\rho}||$.  Then:

\begin{align*}
T &= \frac{1}{p}||(1-Q)\ket{f_0}\bra{f_1} + (1-Q)\ket{f_1}\bra{f_0} + Q\ket{f_2}\bra{f_3} + Q\ket{f_3}\bra{f_2}||\\
&\le \frac{1-Q}{2p_a}||\ket{f_0}\bra{f_1} + \ket{f_1}\bra{f_0}|| + \frac{Q}{2p_a}||\ket{f_2}\bra{f_3} + \ket{f_3}\bra{f_2}||.
\end{align*}

Let $\sigma_0 := \ket{f_0}\bra{f_1} + \ket{f_1}\bra{f_0}$ and $\sigma_1 := \ket{f_2}\bra{f_3} + \ket{f_3}\bra{f_2}$; we must bound $||\sigma_i||$.  By our symmetry assumptions which, as mentioned earlier, forces $Re\braket{f_0|f_1} = Re\braket{f_2|f_3} = 0$, it is clear that $tr\sigma_0 = tr\sigma_1 = 0$ and, so, Lemma \ref{lemma:eigenvalues} applies.  Thus:

\begin{align}
T &\le \frac{1-Q}{p_a}\sqrt{\braket{f_0|f_0}\braket{f_1|f_1}} + \frac{Q}{p_a}\sqrt{\braket{f_2|f_2}\braket{f_3|f_3}}.\label{eq:first-trace-bound}
\end{align}
This, combined with Equation \ref{eq:bound-info}, proves Equation \ref{eq:first-bound}.  Equation \ref{eq:second-bound} follows immediately by observing that $\braket{f_i|f_i} \le 1$ for all $i$.
\end{proof}

Thus, to upper-bound $I(A:C)$, $A$ and $B$ must estimate the values $\braket{f_i|f_i}$ which are simply the probabilities that $C$ sends ``$-1$'' in the event one of the four Bell states was received by him.  Of course, since $A$ and $B$ are semi-quantum, they cannot prepare such states and so cannot measure these values directly.  They can, however, estimate them as we demonstrate next section.

\subsection{Key Rate Estimates}

In order to justify the symmetry assumptions used above, we will first consider a specific scenario where we may call the center $C$ ``semi-honest'' (he follows the protocol description from the previous section, always sending $\ket{\Phi^+}$ on step 1, performing the correct measurements, and reporting the correct result, but beyond this, he may do anything he likes) and where each channel is modeled independently by a depolarizing channel with parameter $p$ in the forward direction and $q$ in the reverse.  That is, if $C$ sends the state $\rho$, the joint state, as it arrives to $A$ and $B$ is $(1-p)\rho + \frac{p}{4}I$; likewise, if the joint state leaving $A$ and $B$ is $\sigma$, the state arriving at $C$ is $(1-q)\sigma + \frac{q}{4}I$.  After this specific example, we will consider a ``worst-case'' scenario which lower-bounds the key rate over any possible attack an adversarial $C$ may perform.

For the following analysis, we will relabel the Bell states as: $\ket{\phi_0} = \ket{\Phi^+}$, $\ket{\phi_1} = \ket{\Phi^-}$, $\ket{\phi_2} = \ket{\Psi^+}$ and $\ket{\phi_3} = \ket{\Psi^-}$.  Since we are assuming a semi-honest $C$, he initially sends the state $\rho = \ket{\phi_0}\bra{\phi_0}$.  When it arrives at $A$ and $B$ however, the joint system is in the state:

\begin{equation}
\rho = (1-p)\ket{\phi_0}\bra{\phi_0} + \frac{p}{4}\sum_{i=0}^3\ket{\phi_i}\bra{\phi_i}.
\end{equation}

The probabilities $p_{i,j}$, that $A$ measures $\ket{i}$ and $B$ measures $\ket{j}$ is easily found to be:
\begin{align*}
p_{0,0} = p_{1,1} &= \frac{1}{2} - \frac{p}{4}\\
p_{0,1} = p_{1,0} &= \frac{p}{4}
\end{align*}
Since $\frac{Q}{2} = p_{0,1} = p_{1,0}$ (from the previous section), this implies $Q = \frac{p}{2}$; clearly, $p$ (and thus $Q$) is a parameter that $A$ and $B$ may estimate.  Also, this scenario fits our first requirement for a ``symmetric'' attack which was the assumption used last section.

Next, assume that $A$ and $B$ both reflect.  Then, the state arriving back at $C$ is:

\begin{align}
\rho' &= (1-q)\rho + \frac{q}{4}\sum_{i=0}^3\ket{\phi_i}\bra{\phi_i}\\
&= (1-q)\left[(1-p)\ket{\phi_0}\bra{\phi_0} + \frac{p}{4}\sum_{i=0}^3\ket{\phi_i}\bra{\phi_i}\right] + \frac{q}{4}\sum_i \ket{\phi_i}\bra{\phi_i}.
\end{align}

Since $C$ is semi-honest, he always reports ``$+1$'' if he measures $\ket{\phi_i}\bra{\phi_i}$ for $i \ne 1$.  From the above equation then, it is clear that the probability $p_w$ that he reports ``$-1$'' in the event both $A$ and $B$ reflect is:

\begin{equation}
p_w = \frac{(1-q)p}{4} + \frac{q}{4}.
\end{equation}

This allows $A$ and $B$ to estimate $q$:
\begin{equation}
q = \frac{4p_w - p}{1-p}
\end{equation}

We now have enough to estimate the values $\braket{f_i|f_i}$.  Recall that $\braket{f_i|f_i}$ is simply the probability that $C$ sends ``$-1$'' if the state of the joint system, when it leaves $A$ and $B$, is $\ket{\phi_i}$.  After receiving a state from $A$ and $B$, which passed through the depolarization channel with parameter $q$, $C$ applies a (limited) attack operator.  Since he is semi-honest this attack, along with the effects of the depolarization channel, leaves us with the following identities:

\begin{align*}
\ket{\phi_0}\bra{\phi_0} &\rightarrow \left((1-q) + \frac{3q}{4}\right)\sigma_{+1}^{(0)}\otimes\ket{+1}\bra{+1}_{cl} + \left(\frac{q}{4}\right)\sigma_{-1}^{(0)}\otimes\ket{-1}\bra{-1}_{cl}\\
\ket{\phi_1}\bra{\phi_1} &\rightarrow \left(\frac{3q}{4}\right)\sigma_{+1}^{(1)}\otimes\ket{+1}\bra{+1}_{cl} + \left((1-q) + \frac{q}{4}\right)\sigma_{-1}^{(1)}\otimes\ket{-1}\bra{-1}_{cl}\\
\ket{\phi_2}\bra{\phi_2} &\rightarrow \left((1-q) + \frac{3q}{4}\right)\sigma_{+1}^{(2)}\otimes\ket{+1}\bra{+1}_{cl} + \left(\frac{q}{4}\right)\sigma_{-1}^{(2)}\otimes\ket{-1}\bra{-1}_{cl}\\
\ket{\phi_3}\bra{\phi_3} &\rightarrow \left((1-q) + \frac{3q}{4}\right)\sigma_{+1}^{(3)}\otimes\ket{+1}\bra{+1}_{cl} + \left(\frac{q}{4}\right)\sigma_{-1}^{(3)}\otimes\ket{-1}\bra{-1}_{cl}.
\end{align*}

As in Equation \ref{eq:C-attack}, we have written the effects of the attack and depolarization channels with respect to the Bell basis.  Here $\sigma_{\pm1}^{(i)}$ is a trace one density matrix representing the state of $C$'s system in the event he sends ``$\pm1''$ if the joint state leaving $A$ and $B$ were to be $\ket{\phi_i}$.  Thus, $\braket{f_1|f_1} = (1-q)+\frac{q}{4}$, and $\braket{f_i|f_i} = \frac{q}{4}$ for all other $i = 0, 2,$ and $3$.  Since $p$ and $q$ are both parameters that $A$ and $B$ may estimate, they may, if they assume this attack model, estimate $\braket{f_i|f_i}$.  Furthermore, the reader may readily verify by performing the same computation above for the case when the state leaving $A$ and $B$ is $\ket{i,j}\bra{i,j}$ for $i,j \in \{0,1\}$, that the probability of $C$ sending ``$-1$'' in the event both $A$ and $B$ measure $0$ is equal to the probability of the same message being sent in case $A$ and $B$ both measure $1$; this probability is simply $\frac{1-q}{2} + \frac{q}{4}$.  Similarly for the probability of $C$ sending ``$-1$'' in the event $A$ and $B$'s measurements differ; that probability being $\frac{q}{4}$.  Clearly this scenario fits the symmetry assumptions we made in the last section; thus these assumptions do not penalize a server running the protocol honestly in the presence of channel noise.  Furthermore, we may apply our main security theorem.

Indeed, the above is enough to compute an upper-bound of $I(A:C)$ using Theorem \ref{thm:main}; however we still need to compute $I(A:B)$.  Let $Q_Z$ be the probability that $A$'s measurement result is different from $B$, given that $C$ sends ``$-1$'' (this may be different from $Q = 2p_{0,1} = 2p_{1,0}$.  Clearly this value is:

\begin{equation}\label{eq:Qz}
Q_Z = \frac{Q(\braket{f_2|f_2} + \braket{f_3|f_3})}{2p_a},
\end{equation}
which in this example is:
\begin{equation}
Q_Z = \frac{pq}{8p_a}.
\end{equation}

It is not difficult to show that $I(A:B) = 1-h(Q_Z)$ providing us with enough information now to estimate the key rate $r$.  For example, if $p=q$ (that is the noise in the forward channel is the same as the noise in the reverse), our protocol maintains a positive key rate for $Q \le 0.199$.  See Figure \ref{fig:honest-bound}.  This tolerated noise level is much higher than BB84 (which can withstand up to $.11$ noise without preprocessing \cite{QKD-renner-keyrate} before the key-rate drops to zero).  From Figure \ref{fig:honest-bound-error}, which plots the values of $p_w$, $p_a$, and $Q_Z$ as $Q$ (and thus $p$ and $q$) increase, we see the main reason: $Q_Z$ is very small due to $C$ ``throwing out'' almost all iterations when $A$ and $B$'s measurement results differ (there are of course some iterations when they are not thrown out due to the noise in the return channel).  Thus $I(A:B)$ remains large - i.e., there is little information leaked during error correction.

We now consider the worst-case scenario, where all noise is induced by $C$'s attack and that $C$ himself is fully adversarial.  Let $\ket{\psi_0} = \sum\alpha_{i,j}\ket{i,j}$, for $\alpha_{i,j} \in \mathbb{C}$, be the initial state sent from $C$ on step 1 of the protocol.  Since we restrict our attention to symmetric attacks, it is required that $|\alpha_{0,0}|^2 = |\alpha_{1,1}|^2 = (1-Q)/2$ and $|\alpha_{0,1}|^2 = |\alpha_{1,0}|^2 = Q/2$.  We claim that, without any loss of power to $C$, we may consider, instead, the joint state of the system arriving at $A$ and $B$, to be: $\ket{\tilde{\psi_0}} = \sqrt{1-Q}\ket{\Phi^+} + \sqrt{Q}\ket{\Psi^+}$.  Since in this worst-case scenario we are assuming all noise is due to $C$'s attack, we assume that this is the state that $C$ initially prepares in step (1) of the protocol.

Indeed, assume $C$ sends the original state $\ket{\psi_0}$.  Since we assume the worst-case that all noise is from $C$'s attack, the state of the system arriving back at $C$, after $A$ and $B$'s operation, is the mixed state:

\begin{align*}
\rho &= p_R\ket{\psi_0}\bra{\psi_0} + p_M\left(\sum_{i,j}|\alpha_{i,j}|^2\ket{i,j}\bra{i,j}\right),
\end{align*}
where $p_R$ is the probability that both $A$ and $B$ reflect and $p_M = 1-p_R$.  If, however, $C$ sends $\ket{\tilde{\psi_0}}$, the state of the joint system after $A$ and $B$'s operation is:

\begin{align*}
\tilde{\rho} = p_R\ket{\tilde{\psi}_0}\bra{\tilde{\psi}_0} + p_M&\left(\frac{1-Q}{2}\ket{0,0}\bra{0,0}+\frac{1-Q}{2}\ket{1,1}\bra{1,1}\right.\\
&+\left.\frac{Q}{2}\ket{0,1}\bra{0,1} + \frac{Q}{2}\ket{1,0}\bra{1,0}\right).
\end{align*}

Let $\alpha_{j,k} = e^{i\theta_{j,k}}p_{j,k}$ with $p_{j,k} = \sqrt{(1-Q)/2}$ if $j=k$; otherwise $p_{j,k} = \sqrt{Q/2}$.  Then, $V = \sum_{j,k}e^{i\theta_{j,k}}\ket{j,k}\bra{j,k}$ is a unitary operator and $V\tilde{\rho}V^* = \rho$.  Thus, if there were an advantage to sending the state $\ket{\psi_0}$ in step 1, $C$ may simply send $\ket{\tilde{\psi}_0}$ initially and apply the operator $V$ later when the system returns to him, all without any loss of power or advantage to $C$.  We may therefore assume $C$ sends the state $\ket{\tilde{\psi}_0}$ initially.

Let $U$ be $C$'s unitary attack operator (Equation \ref{eq:C-attack}).  Then, by linearity:
\[
U\ket{\tilde{\psi}_0} = (\sqrt{1-Q}\ket{e_0} + \sqrt{Q}\ket{e_2})\ket{+1}_{cl} + (\sqrt{1-Q}\ket{f_0} + \sqrt{Q}\ket{f_2})\ket{-1}_{cl}.
\]

Denote by $p_w$ the probability that $C$ sends ``$-1$'' in the event that $A$ and $B$ both reflect.  From the above, this value is:
\begin{equation}\label{eq:pw}
p_w = (1-Q)\braket{f_0|f_0} + Q\braket{f_2|f_2} + 2\sqrt{Q(1-Q)}Re\braket{f_0|f_2}.
\end{equation}
Clearly, if $Q$ and $p_w$ are both small, then so must be $\braket{f_0|f_0}$.  If we upper-bound this value, we attain an upper-bound on $I(A:C)$ using Theorem \ref{thm:main}, thus lower-bounding the key rate.  Of course we must bound this value based on statistics which $A$ and $B$ may measure.  In particular, we will use this value $p_w$.

Let $\ket{f_0} = x\ket{\tilde{f}_0}$ where $\braket{\tilde{f}_0|\tilde{f}_0} = 1$ and, without loss of generality, $x \in \mathbb{R}_{\ge 0}$ (any different phase may be absorbed into $\ket{\tilde{f}_0}$).  Clearly $x = \sqrt{\braket{f_0|f_0}}$.  Write $\ket{f_2} = ye^{i\theta}\ket{\tilde{f}_0} + z\ket{\zeta}$ with $y,z \in \mathbb{R}_{\ge 0}$, $\braket{\zeta|\zeta} = 1$, and $\braket{\tilde{f}_0|\zeta} = 0$, assumptions we may make without any loss of generality.  This of course implies that $y^2+z^2 = \braket{f_2|f_2}$ and $xye^{i\theta} = \braket{f_0|f_2}$.  Using this notation, $p_w$ (Equation \ref{eq:pw}) becomes:

\begin{equation}
p_w = (1-Q)x^2 + Q(y^2+z^2) + 2\sqrt{Q(1-Q)}xy\cos\theta.
\end{equation}

Solving for $x$ and taking the larger solution, we find:

\begin{align*}
x &= \sqrt{f_0|f_0} = \frac{-\sqrt{Q(1-Q)}y\cos\theta}{1-Q}\\
&+ \frac{\sqrt{Q(1-Q)y^2\cos^2\theta - Q(1-Q)(y^2+z^2)+p_w(1-Q)}}{1-Q}\\\\
&= \frac{-\sqrt{Q(1-Q)}y\cos\theta}{1-Q}\\
&+ \frac{\sqrt{Q(1-Q)y^2(\cos^2\theta - 1) - Q(1-Q)z^2 + p_w(1-Q)}}{1-Q}.
\end{align*}

It is trivial to show that this function is maximal when $\theta = \pi$, $y=\sqrt{\braket{f_2|f_2}}$ (note that $y$ can be no larger than this), and $z=0$ yielding:
\begin{align}
\sqrt{f_0|f_0} &\le \frac{\sqrt{1-Q}(\sqrt{Q\braket{f_2|f_2}}+\sqrt{p_w})}{1-Q}\label{eq:f0bound-small}\\
&\le \frac{\sqrt{1-Q}(\sqrt{Q}+\sqrt{p_w})}{1-Q}.\label{eq:f0bound-large}
\end{align}

Thus, using Equation \ref{eq:f0bound-large} with Equation \ref{eq:second-bound} from Theorem \ref{thm:main}, the key rate bounded by:

\begin{equation}\label{eq:keyrate-bound-low}
r \ge 1-h(Q_Z) - \frac{1}{p_a}\left(\sqrt{1-Q}\left(\sqrt{Q}+\sqrt{p_w}\right) + Q\right),
\end{equation}
where, as before, $Q_Z$ is the probability that $A$'s measurement result is different from $B$'s in the event $C$ sends ``$-1$''.  The values $Q_Z$, $Q$, $p_w$, and $p_a$ may all be estimated by $A$ and $B$ without difficulty.  So long as $r > 0$, a secure secret key may be distilled.  Figure \ref{fig:bound26} show graphs of this bound as a function of $Q$ assuming $Q_Z = Q$.  In particular, we observe that, if $p_w = Q$ and $p_a = 1/2$, our protocol maintains a positive key-rate for $Q \le .0335$.

Note, however, that $A$ and $B$ may estimate the value of $\braket{f_2|f_2}$ and $\braket{f_3|f_3}$ by estimating the probability that $C$ sends ``$-1$'' given that $A$ and $B$'s measurement results differ.  Indeed, call this probability $p_{-1}(A\ne B)$ then, using our symmetry assumptions and the fact that $\braket{x|x} \ge 0$ for any $x$:
\[
\frac{1}{2}\braket{f_2|f_2} \le \frac{1}{2}(\braket{f_2|f_2} + \braket{f_3|f_3}) = p_{-1}(A\ne B).
\]
Similarly for $\braket{f_3|f_3}$.

For instance, if both $\braket{f_2|f_2}$ and $\braket{f_3|f_3}$ are less than or equal to $Q$ (considering our analysis in the depolarization scenario above, where $\braket{f_2|f_2} = \braket{f_3|f_3} = Q/2$ assuming $p=q$, this seems a reasonable assumption), then, using Equation \ref{eq:first-bound} with $\braket{f_1|f_1} = 1$ gives us a key rate of:
\begin{equation}\label{eq:keyrate-bound-high}
r \ge 1-h(Q_Z) - \frac{1}{p_a}\left(\sqrt{1-Q}\left(Q + \sqrt{p_w}\right) + Q^2\right).
\end{equation}
Figure \ref{fig:bound27} shows a graph of this key rate bound when $Q_Z = Q^2/p_a$ (from Equation \ref{eq:Qz}, this produces an upper-bound for $h(Q_Z)$ assuming $Q \le \sqrt{p_a/2}$).  We see here, the key rate is positive so long as $Q \le .1065$ when $p_a = 1/2$; the key rate is positive for $Q \le .0525$ when $p_a = .3$ (the smallest value of $p_a$ we observed in the depolarization scenario above).  Compare these results with BB84 as already mentioned; also the three state BB84 protocol, which supports a maximal error rate of $.042$ \cite{QKD-BB84-three-state}, and B92 \cite{QKD-B92} which maintains a positive key rate for $Q \le .048$ (assuming no preprocessing) \cite{QKD-renner-keyrate}.

Observe that, as $p_a$ drops, we do not take into account the necessary drop in the values $\braket{f_i|f_i}$.  In practice, the user may better bound these values using Equation \ref{eq:pra} and also the values $p_{-1}(A \ne B)$ and $p_{-1}(A=B)$ (the latter being the probability that $C$ sends ``$-1$'' given that both $A$ and $B$'s measurement result is the same).  Then, perhaps using numerical optimization, the user may produce a more accurate key rate value.  Our bounds that we described above are exact equations, and are ``worst-case'' (the key rate can only be higher or equal to the values we found).  However, even with $p_a = .3$, using our lower-bounds, the key rate is positive for $Q \le .0525$ which is superior to B92 and three state BB84, and so we are content to leave our analysis at that.  There may be future work in deriving a better lower bound for the key rate taking better account of $p_a$.

\section{Security Against General Attacks and Third-Party Eavesdroppers}

In the previous section we considered only collective attacks.  However, $A$ and $B$ may symmetrize their raw key by permuting it using a permutation chosen publicly by $A$.  As shown in \cite{QKD-general-attack,QKD-general-attack2}, for such permutation invariant protocols, to prove security against general attacks, it is sufficient to consider only collective attacks.  Thus, in the asymptotic scenario, our security key rate bound derived last section applies even against any possible general attack the adversarial sever may perform.

Finally, it is evident that any attack a third-party eavesdropper may perform, may also be ``absorbed'' into the server's attack operator $U$ (including an attack whereby $E$ sends a message on the classical channel, impersonating $C$ - for example, such an impersonation attack will only lead to more noise in the $p_w$ parameter, and/or less accepted iterations).  Thus the lower bound on the key rate from Theorem \ref{thm:main} applies even in this case.  Note that, if we trust the server to be ``semi-honest'', and if we model the third party eavesdropper's attack via a depolarization channel (as was done when analyzing the key rate of B92 \cite{QKD-renner-keyrate}), then we may use our analysis above to show our protocol maintains a positive key rate so long as $Q \le .199$ (though in this case, an authenticated channel should be used by $C$ so that $E$ cannot alter the message).

\section{Closing Remarks}

We have designed a new mediated semi-quantum protocol and proven that it may hold equivalent security guarantees to that of a fully quantum QKD protocol.  Indeed, even in the worst case, our protocol can suffer an amount of noise comparable to fully quantum protocols, before the adversary holds too much information that the users must abort.

There are many questions that might provide fruitful future work.  For instance, it would be nice to improve Theorem \ref{thm:main} to take into better account the factor $p_a$; our bound is a ``worst-case'' solution - but there might be improvements here.

From a practical standpoint, there are several issues.  In particular, it would be better to design a single-qubit version of this protocol so that Bell states are not required.  It is not difficult to do so if we allow the qubit to travel from $C$ to $A$, back to $C$, to $B$, and then return to $C$.  If the qubit sent from $C$ initially was $\ket{+} = \frac{1}{\sqrt{2}}\ket{0} + \frac{1}{\sqrt{2}}\ket{1}$, our protocol should work similarly (now $C$ measuring in the $X$ basis and sending ``$-1$'' if he reads $\ket{-}$).  However the security analysis is more difficult in this setting as the qubit passes through $C$ twice; in this case it would seem our trick of delaying an entanglement attack (Section 4.1) no longer holds on the second pass.

\begin{figure}
\includegraphics[width=250pt]{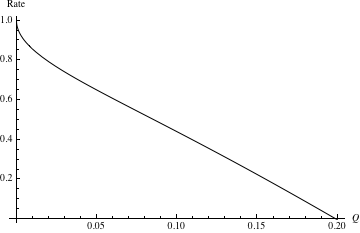}
\caption{Key rate when the center is ``semi-honest'' and the noise in both channels is modeled by two independent depolarizing channels with parameters $p$ and $q$ respectively.  The key rate is plotted as a function of $Q = p/2$: the error rate of $A$ and $B$'s measurement results.  In this graph, we assume $p=q$; that is the noise in both directions is equal.  Observe the key rate is positive so long as $Q \le .199$.}\label{fig:honest-bound}
\end{figure}

\begin{figure}
\includegraphics[width=250pt]{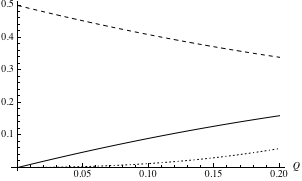}
\caption{A graph of various statistics in the depolarization example when $p=q$ ($Q$ is the error rate of $A$ and $B$'s measurement results; namely $Q = p/2$).  Solid line is the value of $p_w$; dashed line is $p_a$; dotted line (bottom) is $Q_Z$, the actual error in $A$ and $B$'s raw key.  Observe that $Q_Z$ is very small as $Q$ increases.  This is because the quantum server is ``throwing out'' most iterations where their measurement results differ.  This drops the value of $p_a$ and so $A$ and $B$ require more iterations; however less information is leaked due to error correction.}\label{fig:honest-bound-error}
\end{figure}

\begin{figure}
\includegraphics[width=250pt]{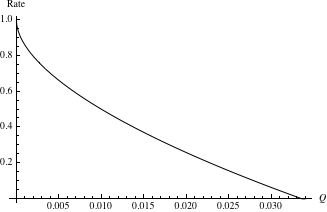}
\caption{A graph of Equation \ref{eq:keyrate-bound-low}, the key rate of our protocol in the worst case scenario when $C$ is adversarial and $A$ and $B$ only use $p_w$ to bound $I(A:C)$.  Here, the key rate is positive so long as $Q \le .0335$.}\label{fig:bound26}
\end{figure}

\begin{figure}
\includegraphics[width=250pt]{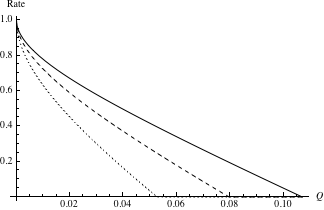}
\caption{A graph of Equation \ref{eq:keyrate-bound-high}, the key rate of our protocol when $C$ is adversarial and $A$ and $B$, using $p_{-1}(A\ne B)$, observe (or enforce) $\braket{f_2|f_2}, \braket{f_3|f_3} \le Q$; also assuming that $p_w = Q$.  Solid line is when $p_a = .5$; dashed line is for $p_a = .4$, dotted line is when $p_a = .3$.  The key rate is positive for $Q \le .1065$ when $p_a = .5$ and when $p_a = .3$, the key rate is positive for $Q \le .0525$;  see the text for a discussion on how $p_a$ might be better accounted for.}\label{fig:bound27}
\end{figure}


\end{document}